\pdfoutput=1
\RequirePackage{ifpdf}
\ifpdf 
\documentclass[pdftex]{sigma}
\else
\documentclass{sigma}
\fi

\usepackage{bm}
\numberwithin{equation}{section}

\newtheorem{Theorem}{Theorem}[section]
\newtheorem{Corollary}[Theorem]{Corollary}
\newtheorem{Lemma}[Theorem]{Lemma}
\newtheorem{Proposition}[Theorem]{Proposition}
 { \theoremstyle{definition}
\newtheorem{Definition}[Theorem]{Definition}
\newtheorem{Remark}[Theorem]{Remark} }

\begin{document}

\allowdisplaybreaks

\newcommand{\arXivNumber}{1808.06748}

\renewcommand{\PaperNumber}{009}

\FirstPageHeading

\ShortArticleName{On Reducible Degeneration of Hyperelliptic Curves and Soliton Solutions}

\ArticleName{On Reducible Degeneration of Hyperelliptic Curves\\ and Soliton Solutions}

\Author{Atsushi NAKAYASHIKI}

\AuthorNameForHeading{A.~Nakayashiki}

\Address{Department of Mathematics, Tsuda University, 2-1-1, Tsuda-Machi, Kodaira, Tokyo, Japan}
\Email{\href{mailto:atsushi@tsuda.ac.jp}{atsushi@tsuda.ac.jp}}

\ArticleDates{Received August 27, 2018, in final form January 29, 2019; Published online February 08, 2019}

\Abstract{In this paper we consider a reducible degeneration of a hyperelliptic curve of genus $g$. Using the Sato Grassmannian we show that the limits of hyperelliptic solutions of the KP-hierarchy exist and become soliton solutions of various types. We recover some results of Abenda who studied regular soliton solutions corresponding to a reducible rational curve obtained as a degeneration of a hyperelliptic curve. We study singular soliton solutions as well and clarify how the singularity structure of solutions is reflected in the matrices which determine soliton solutions.}


\Classification{37K40; 37K10; 14H70}

\section{Introduction}

By the study of \cite{ChK1,ChK2,Kodama,KW} soliton solutions of the KP equation acquire a new aspect. Namely it is discovered that the shapes of soliton solutions are more various than what is known before and those shapes are classified by points of totally positive Grassmannians. This study relates soliton solutions to other areas of mathematics such as cluster algebras.

Then it is natural to ask what happens for quasi-periodic solutions. From this point of view it is important to study the connection of quasi-periodic solutions and soliton solutions, in other words, the degenerations of quasi-periodic solutions to soliton solutions. In papers \cite{Abenda, AbG1, AbG3,AbG2,AbG4} Abenda and Grinevich studied this problem. They constructed a singular rational curve and some divisor on it to each regular soliton solution studied in \cite{ChK1,ChK2, Kodama,KW}. It is noteworthy that their rational curves are reducible in general. It means that we need to consider reducible degenerations of algebraic curves in order
 to obtain a variety of soliton solutions.

In \cite{Abenda} Abenda studied a reducible rational curve which is obtained as a degeneration of a hyperelliptic curve and the corresponding soliton solutions as a concrete example of their theory. It should be noticed that in papers \cite{Abenda, AbG1, AbG3,AbG2,AbG4} soliton solutions and rational curves are directly related and that the limits of quasi-periodic solutions are not actually computed.

We began the study of degenerations of quasi-periodic solutions of the KP-hierarchy by the method of the Sato Grassmannian in \cite{N4}. In this approach it is possible to calculate the limits of quasi-periodic solutions without knowing the limits of periods of a Riemann surface.

In this paper we continue this study. We compute the limit of the $\tau$-function of the KP-hierarchy corresponding to a hyperelliptic curve when it degenerates to a reducible rational curve. From the view point of taking a limit of a solution there is no reason to restrict ourselves to regular solutions. So we consider singular solutions as well. We can see how the singularity structure of the solution is reflected in the matrix $A=(a_{i,j})$ (see Section~\ref{section3}) which determines a~soliton solution.

Consider the hyperelliptic curve $X$ of genus $g=n-1$ given by
\begin{gather*}
y^2=\prod_{j=1}^{2n}(x-\lambda_j).
\end{gather*}
We assume that $\lambda_i$'s are real and ordered as
\begin{gather*}
\lambda_1<\cdots<\lambda_n.
\end{gather*}
There are two points over $x=\infty$ on $X$ which are denoted by $\infty_\pm$. The solution corresponding to~$X$ is well known. It is constructed by the method of Baker--Akhiezer function of Krichever~\cite{Kr}. To construct the Baker--Akhiezer function we need to specify a base point $p_\infty$, a local coordinate~$z$ around $p_\infty$ and a general divisor of degree~$g$. We take $p_\infty=\infty_+$, $z=x^{-1}$. For each $0\leq m_0\leq g$ we consider a general divisor of the form
\begin{gather*}
D_g=p_1+\cdots+p_{m_0}+(g-m_0)\infty_+, \qquad p_j\neq \infty_+ \qquad \forall\, j.
\end{gather*}
The number $m_0$ specifies the partition of the Schur function which appears as the first term in the Schur function expansion of the $\tau$-function corresponding to $D_g$.

Let $k$ be an integer such that $0\leq k\leq m_0$. We assume that $p_1,\dots ,p_k$ is in a small neiborhood of $\infty_-$ and the remaining points are in a small neighborhood of $\infty_+$. The number $k$ specifies the type of soliton solutions in the limit.

We consider the degeneration of $X$ to the reducible curve given by
\begin{gather*}
y^2=\prod_{j=1}^n(x-\lambda_j)^2.
\end{gather*}
To take the limit of the corresponding solution of the KP-hierarchy we use the Sato Grassmannian. Using the Sato Grassmannian it is possible to write down the solution corresponding to~$X$ as a series with the coefficients in the polynomials of $\{\lambda_j\}$. Therefore the limit of the solution exists. By making an appropriate gauge transformation we identify this limit with a soliton solution. For regular solutions $m_0$ must be $n-1$. In this case the soliton solutions obtained here coincide with those in~\cite{Abenda}.

The paper is organized as follows. In Section~\ref{section2} we review the correspondence between solutions ($\tau$-functions) of the KP-hierarchy and points of the Sato Grassmannian. We recall $(n,k)$ solitons and the corresponding points of the Sato Grassmannian in Section~\ref{section3}. In Section~\ref{section4} we review how the data of algebraic curves are embedded in the Sato Grassmannian. In order to embed the data of $X$ to the Sato Grassmannian we need an explicit description of meromorphic functions on $X$ with a pole only at $\infty_+$. It is given in Section~\ref{section5}. We also compute the gap sequence at $\infty_+$ of the holomorphic line bundle of degree $0$ corresponding to the divisor \mbox{$D_g-g\infty_+$}. The top term of the Schur function expansion of the solution is determined by using it. In Section~\ref{section6} we recall the description of the tau function corresponding to $D_g$ in terms of Riemann's theta function. The limit of the frame of the Sato Grassmannian corresponding to~$D_g$ is determined in Section~\ref{section7}. We show that it is gauge equivalent to the frame of an $(n,k+1)$ soliton. Finally we give the explicit formula of the limits of the tau function and the adjoint wave function (dual Baker--Akhiezer function) in Section~\ref{section8}.

\section{Sato Grassmannian}\label{section2}
\subsection{KP-hierarchy}
We set
\begin{gather*}
[w]={\vphantom{\bigg(}}^t\left(w,\frac{w^2}{2},\frac{w^3}{3},\dots \right).
\end{gather*}
In this paper the KP-hierarchy signifies the following equation \cite{DJKM} for the function
$\tau(x)$ of $x={}^t(x_1,x_2,x_3,\dots )$:
\begin{gather}
\oint {\rm e}^{-2\sum\limits_{j=1}^\infty y_j\lambda^j}\tau\big(x-y-\big[\lambda^{-1}\big]\big)\tau\big(x+y+\big[\lambda^{-1}\big]\big)\frac{{\rm d}\lambda}{2\pi {\rm i}}=0,\label{KP-hierarchy}
\end{gather}
where $y={}^t(y_1,y_2,y_3,\dots )$ and the integral means taking the coefficient of $\lambda^{-1}$ in the series expansion of the integrand in $\lambda$.

If we set $u=2\partial_{x_1}^2\log \tau(x)$, it satisfies the KP equation
\begin{gather}
3u_{x_2x_2}+(-4u_{x_3}+6uu_{x_1}+u_{x_1x_1x_1})_{x_1}=0.\label{KP-equation}
\end{gather}

\subsection{Sato Grassmanian}
The set of formal power series solutions of the KP-hierarchy is parametrized by the Sato Grassmannian which we denote by UGM~\cite{S,SN} (see also \cite{KNTY, Mul}). Let us briefly recall the definition and the fundamental properties of UGM.

Let $V={\mathbb C}((z))$ be the vector space of formal Laurent series in the variable $z$ and $V_\phi={\mathbb C}\big[z^{-1}\big]$, $V_0=z{\mathbb C}[[z]]$ subspaces of $V$. Then we have
\begin{gather*}
V=V_\phi \oplus V_0, \qquad V/V_0\simeq V_\phi.
\end{gather*}
Let $\pi\colon V\rightarrow V_\phi$ be the projection map. Then UGM is the set of subspaces $U$ of $V$ which satisfy
\begin{gather*}
 \dim \operatorname{Ker}(\pi|_U)=\dim \operatorname{Coker}(\pi|_U)<\infty.
\end{gather*}

A basis of $U$ is called a frame of $U$. We express a frame of $U$ by an infinite matrix as follows. Set
\begin{gather*}
e_i=z^{i+1},\qquad i\in {\mathbb Z},
\end{gather*}
and write an element $f$ of $V$ as
\begin{gather*}
f=\sum_{i\in {\mathbb Z}} \xi_i e_i.
\end{gather*}

We associate the infinite column vector $(\xi_i)_{i\in{\mathbb Z}}$ to $f$. Then a frame of $U$ is given by a matrix $\xi=(\xi_{i,j})_{i\in {\mathbb Z},j\in {\mathbb N}}$ which is written as
\begin{gather*}
\xi=\left(\begin{matrix}
\quad&\vdots&\vdots\\
\cdots&\xi_{-2,2}&\xi_{-2,1}\\
\cdots&\xi_{-1,2}&\xi_{-1,1}\\
---&---&---\\
\cdots&\xi_{0,2}&\xi_{0,1}\\
\cdots&\xi_{1,2}&\xi_{1,1}\\
\quad&\vdots&\vdots\\
\end{matrix}\right).
\end{gather*}

For a point $U$ of UGM there exists a frame $\xi=(\xi_{i,j})_{i\in {\mathbb Z},j\in {\mathbb N}}$ satisfying the following conditions: there exists a non-negative integer $l$ such that
\begin{gather}
\xi_{i,j}=
\begin{cases}
1&\text{if $j>l$ and $i=-j$},\\
0&\text{if ($j>l$ and $i<-j$) or ($j\leq l$ and $i<-l$)}.
\end{cases} \label{frame-cond}
\end{gather}

It means that $X$ is of the form
\begin{gather*}
\xi=\left[\begin{matrix}
\ddots&\quad&O&\quad\\
\cdots&1&\quad&\quad\\
\cdots&\ast&1&\quad\\
\cdots&\ast&\ast&B\\
\end{matrix}
\right],
\end{gather*}
where $B$ is an $\infty\times l$ matrix of rank $l$ and its first row is placed at the $-l$th row of $\xi$. Conversely a matrix of this form becomes a frame of a point of UGM. In the following a frame of a point of UGM is always assumed to satisfy the condition~(\ref{frame-cond}) unless otherwise stated.

Here we introduce the notion of Maya diagram. A Maya diagram of charge $p$ is a sequence of integers $M=(m_1,m_2,\dots )$ such that $m_1>m_2>\cdots$ and, for some $l$, $m_i=-i+p$, $i\geq l$ holds. In this paper we consider only a Maya diagram of charge $0$ and call them simply a Maya diagram.

With each Maya diagram $M$ we can associate the partition $\lambda$ by
\begin{gather*}
\lambda=(m_1+1,m_2+2,\dots ).
\end{gather*}
This gives a one to one correspondence between the set of Maya diagrams and the set of partitions.

Let $\lambda=(\lambda_1,\dots ,\lambda_l)$ be an arbitrary partition and $M=(m_1,m_2,m_3,\dots )$ the corresponding Maya diagram. The Pl\"ucker coordinate $\xi_\lambda$ or $\xi_M$ of a frame $\xi$ is defined by
\begin{gather*}
\xi_\lambda=\xi_M=\det(\xi_{m_i,j})_{i,j\in {\mathbb N}}.
\end{gather*}
We introduce the Schur function $s_\lambda(x)$ of the variable $x={}^t(x_1,x_2,\dots )$ by
\begin{gather*}
s_\lambda(x)=\det(p_{\lambda_i-i+j}(x))_{1\leq i, j\leq l},\qquad
\exp\left(\sum_{i=1}^\infty x_i\lambda^i\right)=\sum_{i=0}^\infty p_i(x) \lambda^i.
\end{gather*}

Then we define the tau function corresponding to a frame $\xi$ of a point of UGM by
\begin{gather*}
\tau(x;\xi)=\sum_\lambda \xi_\lambda s_\lambda(x),
\end{gather*}
where the summation is taken over all partitions.

For a given point of UGM a frame $\xi$ of it satisfying the condition (\ref{frame-cond}) is not unique. If $\xi$ is replaced by another frame the tau function is multiplied by a non-zero constant.

\begin{Theorem}[\cite{SS}]\label{Sato1} For a frame $\xi$ of a point of UGM $\tau(x;\xi)$ is a solution of the KP-hierarchy. Conversely for any formal power series solution $\tau(x)$ of the KP-hierarchy there exists a unique point $U$ of ${\rm UGM}$ and a frame $\xi$ of $U$ such that $\tau(x)=\tau(x;\xi)$.
\end{Theorem}

\section[$(n,k)$ solitons]{$\boldsymbol{(n,k)}$ solitons}\label{section3}
In this section we recall the results on $(n,k)$ solitons (see \cite{Kodama} for more details).

For a positive integer $N$ and a nonnegative integer $N'$ we use the following notation:
\begin{gather*}
[N]=\{1,\dots ,N\},
\qquad \binom{[N]}{N'}=\big\{(i_1,\dots ,i_{N'})\in [N]^{N'}\,|\, i_1<\cdots<i_{N'}\big\}.
\end{gather*}

Let $n$, $k$ be positive integers which satisfy $n\geq k$, $A=(a_{ij})$ be an $n\times k$ matrix of rank $k$ and $\lambda_1,\dots ,\lambda_n$ non-zero complex numbers.

For $I=(i_1,\dots ,i_k) \in \binom{[n]}{k}$ we set
\begin{gather*}
A_I=\det(a_{i_p,q})_{1\leq p,q\leq k},\qquad
\Delta_I(\lambda_1,\dots ,\lambda_n)=\prod_{p<q}(\lambda_{i_q}-\lambda_{i_p}).
\end{gather*}
Then
\begin{gather}
\tau(x)=\sum_{I\in \binom{[n]}{k}} \Delta_I(\lambda_1,\dots ,\lambda_n)A_I \exp\left(\sum_{i\in I}\eta_i\right),\qquad \eta_i=\sum_{j=1}^\infty x_j\lambda_i^j\label{nk-soliton}
\end{gather}
becomes a solution of the KP-hierarchy \cite{FN,Satsuma2}. It is called the $(n,k)$ soliton associated with the data $(A,\{\lambda_j\})$ or the $(n,k)$ soliton associated with $A$ if $\{\lambda_j\}$ are fixed.

The $(n,k)$ soliton (\ref{nk-soliton}) can be written in the form of Wronskian. Let
\begin{gather*}
S_j=\sum_{i=1}^n a_{ij} \exp(\eta_i).
\end{gather*}
Then
\begin{gather*}
\tau(x)=\operatorname{Wr}(S_1,\dots ,S_k)=\det\big(S^{(i-1)}_j\big)_{1\leq i,j\leq k}, \qquad S^{(i)}=\frac{\partial ^i S}{\partial x_1^i}.
\end{gather*}

\begin{Remark} In the case $n=k$, $\tau(x)=C\exp\Big(\sum\limits_{i=1}^\infty d_i x_i\Big)$ for some constants $C$, $d_i$. It is a trivial solution of~(\ref{KP-hierarchy}) which is obtained from the constant solution by a~gauge transformation. We include this case for the sake of convenience to describe the limits of the quasi-periodic solutions later.
\end{Remark}

The point of UGM corresponding to an $(n,k)$ soliton is determined by Sato~\cite{S}. We consider the function $1/(1-\lambda_i z)$ as a power series in $z$ by
\begin{gather*}
\frac{1}{1-\lambda_i z}=\sum_{r=0}^\infty \lambda_i^r z^r.
\end{gather*}
Then

\begin{Theorem}[\cite{S}]\label{S-frame}
The point of UGM corresponding to the $(n,k)$ soliton associated with $(A,\{\lambda_j\})$ is given by the following frame:
\begin{gather}
z^{-(k-1)}\sum_{i=1}^n \frac{a_{ij}}{1-\lambda_i z}, \qquad j\in [k], \qquad z^{-j}, \qquad j\geq k.\label{soliton-frame}
\end{gather}
\end{Theorem}

\section{Algebraic curves and UGM}\label{section4}
It is possible to embed certain set of data of algebraic curves to the Sato Grassmannian (see \cite{KNTY,Mul,SW} and the references therein). We restrict ourselves to the sepecial case which is relevant to us.

Let $X$ be a compact Riemann surface of genus $g$, $p_\infty$ a point of $X$, $z$ a local coordinate of~$X$ around~$p_\infty$, $L$ a holomorphic line bundle of degree $g-1$ and $\phi$ a local trivialization of~$L$ around~$p_\infty$. We define a map
\begin{gather*}
\iota\colon \ H^0(X,L(\ast p_\infty))\longrightarrow V
\end{gather*}
as follows. Take an element $s$ of $H^0(X,L(\ast p_\infty))$. Using $\phi$ the section $s$ can be considered as a~meromorphic function on some neighborhood of $p_\infty$. Therefore it is possible to expand it in~$z$ as
\begin{gather*}
\phi(s)=\sum_{n=-\infty}^{+\infty} s_nz^n.
\end{gather*}
Define
\begin{gather*}
\iota(s)=\sum_{n=-\infty}^{+\infty}s_n e_n=\sum_{n=-\infty}^{+\infty} s_nz^{n+1}.
\end{gather*}

Then

\begin{Theorem}[\cite{KNTY,Mul,SW}] The image of $\iota$ belongs to ${\rm UGM}$.
\end{Theorem}

Let us interpret this theorem in terms of dvisors and meromorphic functions.

Let $m_0$ be an integer satisfying $0\leq m_0\leq g$, $p_j$, $j\in [m_0]$, points of $X$ such that $p_j\neq p_\infty$ for any $j$, $D=p_1+\cdots+p_{m_0}+(g-1-m_0)p_\infty$ the divisor of degree $g-1$ and $L$ the holomorphic line bundle corresponding to~$D$. Then $L\simeq {\cal O}(D)$ as a sheaf of ${\cal O}$-modules. Using this isomorphism and the local coordinate $z$ we can consider a~local section of $L$ near $p_\infty$ as a meromorphic function on some neighborhood of $p_\infty$. It gives a local trivialization of $L$ around $p_\infty$. So let us examine how this isomorphism looks like.

Let $I$ be a finite index set which contains the symbol $\infty$, $\{W_i\,|\,i\in I\}$ an open covering of~$X$ such that each $W_i$ is a domain of a local coordinate system of $X$ and contains at most one~$p_j$ and~$d_i$ a meromorphic function on $W_i$ whose divisor is $D$ in $W_i$. We assume that $W_\infty$ contains~$p_\infty$. We can take $d_\infty=z^{g-1-m_0}$. Then $d_{jk}=d_j/d_k$ defines a transition function of the line bundle $L$. Let $W$ be an open set and $\{(s_j,W_j)\}$ a local holomorphic section of $L$ over $W$. It means that, if $W\cap W_j\cap W_k$ is not empty, $s_j=d_{jk}s_k$ on $W\cap W_j\cap W_k$. Then $s_j/d_j=s_k/d_k$ on $W\cap W_j\cap W_k$. Therefore $f=\{(s_j/d_j,W_j)\}$ defines a meromorphic function on $W$ whose divisor $(f)$ satisfies $(f)+D\geq 0$. This is the map from $L$ to ${\cal O}(D)$.

Let us look at the neighborhood $W_\infty$ of $p_\infty$. A local section $s$ of $L$ on $W_\infty$ is mapped to the meromorphic function $s/z^{g-1-m_0}$ on $W_\infty$. Conversely a local meromorphic function $f$ on $W_\infty$ which belongs to ${\cal O}(D)$ corresponds to the local holomorphic section $s=z^{g-1-m_0}f$ of $L$.

We have the composition of maps:
\begin{gather*}
\tilde{\iota}\colon \ H^0(X,{\cal O}(D+\ast p_\infty)) \longrightarrow H^0(X,L(\ast p_\infty)) \longrightarrow V,
\end{gather*}
where the first map is that induced from ${\cal O}(D)\simeq L$ and the second map is $\iota$. Using the description of the isomorphism ${\cal O}(D)\simeq L$ explained above $\tilde{\iota}$ is given as follows.

Let us take a meromorphic function $f\in H^0(X,{\cal O}(D+\ast p_\infty))$ and expand it in $z$ around $p_\infty$ as
\begin{gather*}
f=\sum f_n z^n.
\end{gather*}
Then
\begin{gather*}
\tilde{\iota}(f)=\iota\left(z^{g-1-m_0}\sum f_n z^{n}\right)= \sum f_n z^{n+g-m_0}=z^{g-m_0}f.
\end{gather*}

\begin{Corollary}\label{image-tilde-iota}
The subspace ${\tilde{\iota}\big(H^0(X,{\cal O}(D+\ast p_\infty))\big)}$ belongs to ${\rm UGM}$.
\end{Corollary}

\section{Hyperelliptic curves and functions on them}\label{section5}
Let $X$ be the hyperelliptic curve of genus $g=n-1$ defined by
\begin{gather}
y^2=\prod_{j=1}^{2n}(x-\lambda_j),\label{hyperelliptic}
\end{gather}
where $\{\lambda_j\}$ are mutually distinct non-zero complex numbers. It can be compactified by adding two points over $x=\infty$ which we denote by $\infty_\pm$. We take $z=1/x$ as a local coordinate around~$\infty_\pm$. We distinguish $\infty_+$ and $\infty_-$ by the expansion of $y$:
\begin{gather*}
y=\pm z^{-n} (1+O(z) ) \qquad \text{at $\infty_\pm$.}
\end{gather*}
We denote by $\sigma$ the involution of $X$ defined by $\sigma(x,y)=(x,-y)$.

Let
\begin{gather}
p_1+\cdots+p_g, \qquad p_j\in X,\label{general-divisor}
\end{gather}
be a general divisor. It is known that (\ref{general-divisor}) is a general divisor if and only if $p_i\neq\sigma(p_j)$ for any $i\neq j$ (see \cite{Fay} for example). Let
\begin{gather*}
D=p_1+\cdots+p_g-\infty_+
\end{gather*}
the divisor of degree $g-1$.

 It can be written as
\begin{gather}
D=p_1+\cdots+p_{m_0}+(g-m_0-1)\infty_+, \qquad p_j\neq \infty_+, \qquad j\in[m_0].\label{D-2}
\end{gather}
for some $0\leq m_0\leq g$. Since (\ref{general-divisor}) is a general divisor,
\begin{gather}
p_i\neq \sigma(p_j) \qquad i\neq j, \label{cond-general}\\
p_j\neq \infty_-, \qquad j\in[m_0] \quad \text{if $m_0<g$}.\nonumber
\end{gather}

For simplicity we assume that $p_1,\dots ,p_{m_0}$ are mutually distinct and different from $\infty_-$.

Let us find a basis of $H^0(X,{\cal O}(D+\ast\infty_+))$. To this end we first study the case of $m_0=0$, that is, the case $D=(g-1)\infty_+$. In this case
\begin{gather*}
H^0(X,{\cal O}(D+\ast\infty_+))=H^0(X,{\cal O}(\ast \infty_+)),
\end{gather*}
where the right hand side is the space of meromorphic functions on $X$ which are holomorphic on $X\backslash \{\infty_+\}$. A basis of this space can be given as follows.

It can be easily proved that the space of meromorphic functions on $X$ which are holomorphic on $X\backslash\{\infty_+,\infty_- \}$ is equal to the space of polynomials in $x$ and $y$. Let us write the expansion of~$y$ at~$\infty_\pm$ as
\begin{gather}
y=\pm z^{-n}\left(\prod_{j=1}^{2n}(1-\lambda_j z)\right)^{1/2}=\pm z^{-n}\sum_{j=0}^\infty \alpha_j z^j, \qquad \alpha_0=1.\label{y-expand}
\end{gather}
For $m\geq n$ define polynomials $g_m(x)$ and $f_m(x,y)$ by
\begin{gather*}
g_m(x)=\sum_{j=0}^{m}\alpha_j x^{m-j}, \qquad f_m(x,y)=\frac{1}{2}\big(x^{m-n}y+g_m(x)\big).
\end{gather*}
Since, at $\infty_\pm$,
\begin{gather*}
g_m(x)=z^{-m}\sum_{j=0}^m \alpha_j z^j,
\end{gather*}
we have
\begin{gather}
f_m=\begin{cases}
z^{-m}(1+O(z))&\text{at $\infty_+$},\\
O(z)& \text{at $\infty_-$}.
\end{cases} \label{exp-fm}
\end{gather}
This means that, for $m\geq n$, $f_m$ is a meromorphic function on $X$ with a pole only at $\infty_+$ and the order of a pole is $m$.

Here we recall the notion of gaps. Let $M$ be a holomorphic line bundle of degree zero, $p$~a~point of $X$ and $m$~a non-negative integer. If there is no meromorphic section of $M$ with a~pole of order $m$ at $p$ and with no other poles, then $m$ is called a gap of~$M$ at~$p$. If $m$ is not a~gap then it is called a non-gap of~$M$ at~$p$. There are exactly $g$ gaps for any~$M$ and~$p$ by \cite[Lemma~1]{N3}. If the set of gaps of the trivial line bundle at $p$ is not $[g]$, then $p$ is called a Weierstrass point. It is known that the Weierstrass points of the hyperelliptic curve $X$ are branch points $(\lambda_j,0)$, $j\in [2g]$. In particular $\infty_\pm$ are not Weierstrass points.

\begin{Lemma} The functions $\{1, f_m, m\geq n\}$ is a basis of the vector space $H^0(X,{\cal O}(\ast \infty_+))$.
\end{Lemma}

\begin{proof}Since $\infty_+$ is not a Weierstrass point, the gaps at $\infty_+$ is $1,2,\dots ,g$. Thus $H^0(X,{\cal O}(\ast \infty_+))$ is generated, as a vector space, by $f_m$, $m\geq n=g+1$ and~$1$. The linear independence follows from the expansion (\ref{exp-fm}) at~$\infty_+$.
\end{proof}

Next we consider the general case (\ref{D-2}) with $m_0$ not necessarily equal to zero. Since $p_i\neq \infty_\pm$, we can write
\begin{gather*}
p_i=(c_i,y_i),
\end{gather*}
for some $c_i\in {\mathbb C}$. We assume that $c_i$ does not depend on $\{\lambda_j\}$ for any $i$. In particular $c_i\neq \lambda_j$, $i,j\in [2n]$. Since $\{p_i\}$ are mutually distinct and satisfy~(\ref{cond-general}), $\{c_i\}$ are mutually distinct. In the following we assume further $c_j\neq 0$, $j\in [m_0]$.

For $j\in [m_0]$ define
\begin{gather}
h_j=\frac{f_n(x,y)-f_n(c_j,-y_j)}{x-c_j}=\frac{y+g_n(x)-(-y_j+g_n(c_j))}{2(x-c_j)}.\label{hj}
\end{gather}
It is a meromorphic function on $X$ with the pole divisor $p_j+(n-1)\infty_+$.

\begin{Lemma}\label{fm-hj} The functions $\{1, f_m, m\geq n, h_j , j\in [m_0]\}$ is a basis of $H^0 (X,{\cal O}(D+\ast\infty_+) )$
\end{Lemma}
\begin{proof} Let $M$ be the holomorphic line bundle of degree zero corresponding to the divisor $p_1+\cdots+p_{m_0}-m_0 \infty_+$,
\begin{gather}
M\simeq {\cal O}(p_1+\cdots+p_{m_0}-m_0 \infty_+).\label{M}
\end{gather}
Then we have
\begin{gather}
M(m\infty_+)\simeq{\cal O}(p_1+\cdots+p_{m_0}+(m-m_0 )\infty_+),\label{flat-M}\\
H^0(X, M(\ast \infty_+))\simeq H^0(X, {\cal O}(D+\ast\infty_+)).\label{MD-correspondence}
\end{gather}
We identify the left hand side of (\ref{MD-correspondence}) with the right hand side of (\ref{MD-correspondence}). Then $1$ and $h_j$, $ j\in [m_0]$, belong to $H^0(X, M((m_0+g)\infty_+))$. Since $c_1,\dots ,c_{m_0}$ are mutually distinct, the set of functions $\{1, h_j, j\in [m_0]\}$ is linearly independent and it spans an $m_0+1$ dimensional subspace of $H^0(X, M((m_0+g)\infty_+))$. Since the degree of $M$ is zero
\begin{gather*}
\dim H^0(X, M((m_0+g)\infty_+))\leq m_0+g+1,
\end{gather*}
Notice that, for $m\geq g+1$,
\begin{gather*}
f_m\in H^0(X, M((m_0+m)\infty_+)), \qquad f_m\notin H^0(X, M((m_0+m-1)\infty_+)).
\end{gather*}
Therefore there are at most $m_0+g+1-(m_0+1)=g$ gaps in $H^0(X, M(\ast \infty_+))$. Since there are exactly $g$ gaps by \cite[Lemma~1]{N3}, we can conclude that $\{1,h_j, j\in [m_0]\}$ is a basis of $H^0(X,M((m_0+g)\infty_+))$. It then shows that $\{1, h_j, j\in [m_0], f_m, m\geq g+1\}$ is a basis of $H^0(X, M(\ast \infty_+))$.
\end{proof}

Let us determine the gap sequence of $M$ defined by (\ref{M}) at $\infty_+$. By (\ref{flat-M}) a meromorphic function from $H^0(X,{\cal O}(D+\ast\infty_+))$ with a pole of order $r$ at $\infty_{+}$ is identified with a meromorphic section of $M$ with a pole of order $r+m_0$ at $\infty_+$. We prove

\begin{Proposition}\label{gap-sequence} The gap sequence of $M$ at $\infty_+$ is $ (0,1,\dots ,m_0-1,m_0+1,\dots ,g)$.
 \end{Proposition}

Let $K=(k_{ij})_{1\leq i,j\leq m_0}$ be the $m_0\times m_0$ matrix defined by
 \begin{gather*}
 k_{ij}=\sum_{s=0}^{j-1}\alpha_s c_i^{j-1-s}.
\end{gather*}

 \begin{Lemma}\label{basis-change}\quad
 \begin{enumerate}\itemsep=0pt
\item[$(i)$] $\det K=\prod\limits_{1\leq i<j\leq m_0} (c_j-c_i)$.
\item[$(ii)$] Let $K^{-1}=(k'_{ij})_{1\leq i,j\leq m_0}$ and
 \begin{gather*}
 \tilde{h}_i(z)=\sum_{j=1}^{m_0} k'_{ij} h_j(z). 
 \end{gather*}
 Then
 \begin{gather*}
 \tilde{h}_i(z)=z^{-(g+1-i)}+O\big(z^{-(g-m_0)}\big), \qquad 1\leq i\leq m_0. 
 \end{gather*}
 \end{enumerate}
 \end{Lemma}

\begin{proof} (i) It can be proved just by computation using the properties of determinants. So we leave the details to the reader.

 (ii) By expanding $h_i(z)$ in $z$ we have
 \begin{gather*}
 h_i(z)=z^{-g}\left(\sum_{j=1}^{m_0}k_{ij} z^{j-1}+O\big(z^{m_0}\big)\right).
 \end{gather*}
The assertion (ii) follows from this.
\end{proof}

By the lemma we have

\begin{Corollary}\label{basis-tilde-h}\quad\samepage
\begin{enumerate}\itemsep=0pt
\item[$(i)$] The following set of functions is a basis of $H^0(X,{\cal O}(D+\ast\infty_+))$,
\begin{gather*}
1, \quad f_m, \quad m\geq n, \quad \tilde{h}_i, \quad i\in[m_0].
\end{gather*}
\item[$(ii)$] The expansion coefficients of $f_m(z)$ and $\tilde{h}_i(z)$ are polynomials of $\{\lambda_j\}$.
\end{enumerate}
\end{Corollary}

\begin{proof}(i) The assertion follows from Lemmas~\ref{fm-hj} and~\ref{basis-change}.

(ii) By (\ref{y-expand}) $\alpha_i$ is a polynomial in $\{\lambda_j\}$ for any $i$. Then $k_{ij}$ and the expansion coefficients of~$f_m$ are polynomials of $\{\lambda_j\}$ by their definitions. It follows that the expansion coefficients of~$h_i(z)$ are polynomials of~$\{\lambda_j\}$. Then $k_{ij}'$ is a polynomial of $\{\lambda_r\}$ by Lemma~\ref{basis-change}(i). Consequently the expansion coefficients of $\tilde{h}_i(z)$ are polynomials of~$\{\lambda_j\}$.
\end{proof}

\begin{proof}[Proof of Proposition \ref{gap-sequence}] By Lemma \ref{basis-change}(ii) and Corollary~\ref{basis-tilde-h}(i) we see that $m_0$, $m_0+g+1-i$, $i\in [m_0]$, $m_0+m$, $m\geq g+1$ are non-gaps. The complement of these numbers in non-negative integers consists of $0,1,\dots ,m_{0}-1,m_0+1,\dots ,g$. Since the number of gaps of~$M$ at~$\infty_+$ is~$g$, these $g$ numbers are exactly the gaps.
\end{proof}

\section{Theta function solution}\label{section6}
By Corollary \ref{image-tilde-iota}, Lemma \ref{fm-hj} and Corollary \ref{basis-tilde-h}(i) it is possible to give the following definition.

\begin{Definition}\quad\begin{enumerate}\itemsep=0pt
\item[(i)] Define the point $U(D)$ of UGM by
\begin{gather*}
U(D)=\tilde{\iota}\big(H^0(X,{\cal O}(D+\ast\infty_+))\big).
\end{gather*}
\item[(ii)] Define the frames $\xi(D)$ and $\tilde{\xi}(D)$ of $U(D)$ by
\begin{gather*}
\xi(D) =\big(\ldots,\tilde{\iota}(f_{n+1}), \tilde{\iota}(f_{n}),\tilde{\iota}(h_{m_0}),\ldots,\tilde{\iota}(h_{1}), \tilde{\iota}(1)\big),\\ 
\tilde{\xi}(D)=\big(\ldots,\tilde{\iota}(f_{n+1}), \tilde{\iota}(f_{n}),\tilde{\iota}\big(\tilde{h}_{m_0}\big),\ldots,\tilde{\iota}\big(\tilde{h}_{1}\big), \tilde{\iota}(1)\big).
\end{gather*}
\end{enumerate}
\end{Definition}

By Lemma \ref{basis-change} the tau functions corresponding to $\xi(D)$ and $\tilde{\xi}(D)$ are related by
\begin{gather}
\tau(x;\xi(D))=\left(\prod_{1\leq i<j\leq m_0}(c_j-c_i)\right) \tau\big(x;\tilde{\xi}(D)\big).\label{relation-tau}
\end{gather}
By Krichever's construction \cite{Kr} the tau function $\tau\big(x;\tilde{\xi}(D)\big)$ is expressed in terms of Riemann's theta function as follows.

Let $\{\epsilon_i,\delta_i\, |\, i\in[g]\}$ be a canonical homology basis, $\{{\rm d}v_j\,|\, j\in [g]\}$ the normalized holomorphic one forms, $\Omega=\big(\int_{\delta_j} {\rm d}v_i\big)$ the period matrix, $\theta(z\,|\,\Omega)$ Riemann's theta function and $K_{\infty_+}$ Riemann's constant corresponding to the point~$\infty_+$.

For $i\geq 1$ we denote by ${\rm d}\tilde{r}_i$ the normalized differential of the second kind with a~pole only at~$\infty_+$ of order $i+1$. Namely it satisfies
\begin{gather*}
{\rm d}\tilde{r}_i={\rm d}\big(z^{-i}+O(z)\big) \qquad \text{at $\infty_+$},\qquad
\int_{\epsilon_j} {\rm d}\tilde{r}_i=0, \qquad j\in[g].
\end{gather*}

Define $\gamma_{ij}$, $\Gamma$, $e$ by
\begin{gather}
{\rm d}v_i = \sum_{j=1}^\infty \gamma_{ij} z^{j-1} {\rm d}z, \qquad \Gamma=(\gamma_{ij})_{i\in [g], j\geq 1},\nonumber\\
e=-\sum_{j=1}^{m_0}\int_{\infty_+}^{p_j} {\rm d}{\bm v}+K_{\infty_+},\qquad {\rm d}{\bm v}={}^t({\rm d}v_1,\dots ,{\rm d}v_g).\label{def-e}
\end{gather}
Since $p_1+\cdots+p_{m_0}+(g-m_0)\infty_+$ is a general divisor, $\theta(\int_{\infty_+}^p {\rm d}{\bm v}+e)$ has a zero of order $g-m_0$ at $\infty_+$ by Riemann's theorem~\cite{Fay}. Therefore
\begin{gather*}
\theta_0:=\frac{1}{(g-m_0)!}\left( \left(\sum_{j=1}^g\frac{{\rm d} v_j}{{\rm d} z}\frac{\partial}{\partial z_j}\right)^{g-m_0}\theta\right)(e\,|\,\Omega)
\end{gather*}
does not vanish.
By Krichever's theory \cite{Kr} the following function $\Psi(x;z)$ defines an adjoint wave function~\cite{DJKM},
\begin{gather*}
\Psi(x;z)=\frac{z^{g-m_0}\theta_0 \theta\big(\Gamma x+\int_{\infty_+}^p {\rm d}{\bm v}+e\,|\,\Omega\big)}{\theta\big(\int_{\infty_+}^p {\rm d}{\bm v}+e\,|\,\Omega\big)\theta(\Gamma x+e\,|\,\Omega)}
\exp\left(-\sum_{i=1}^\infty \int^p {\rm d}\tilde{r}_i\right),
\end{gather*}
where $\int^p {\rm d}\tilde{r}_i$ is the indefinite integral without the constant term. Let
\begin{gather*}
[z]=\left(z,\frac{z^2}{2},\frac{z^3}{3},\dots \right).
\end{gather*}
By \cite{DJKM} there exists, up to a constant multiple, a function $\tau(x)$ which satisfies the following equation near $\infty_+$,
\begin{gather}
\Psi(x;z)=\frac{\tau(x+[z])}{\tau(x)}\exp\left(-\sum_{i=1}^\infty x_i z^{-i}\right).\label{Psi-tau-relation}
\end{gather}
Since $z^{-(g-m_0)}\Psi(x;z)$ is invariant when $p$ goes round $\epsilon_i$, $\delta_i$ cycles, the expansion coefficients in $x_1,x_2,\dots $ of $z^{-(g-m_0)}\Psi(x;z)\theta(\Gamma x+e\,|\,\Omega)$ are elements of $H^0(X,{\cal O}(D+\ast \infty_+))$. Therefore $\tau(x)$ coincides with $\tau\big(x;\tilde{\xi}(D)\big)$ up to a constant multiple (see~\cite{N2}).

The function $\tau(x)$ satisfying the relation (\ref{Psi-tau-relation}) can be constructed in the following way. Let $E(p_1,p_2)$ be the prime form~\cite{Fay}. Write
\begin{gather*}
E(p_1,p_2)=\frac{E(z_1,z_2)}{\sqrt{{\rm d}z_1}\sqrt{{\rm d}z_2}},
\end{gather*}
where $z_i=z(p_i)$, $i=1,2$.

Define $q_{i,j}$, $\beta_j$, $q(x)$ by
\begin{gather*}
{\rm d}_{z_1}{\rm d}_{z_2}\log E(z_1,z_2) = \left(\frac{1}{(z_1-z_2)^2}+\sum_{i,j\geq 1}q_{ij}z_1^{i-1}z_2^{j-1}\right){\rm d}z_1{\rm d}z_2,\\
\log\left(\frac{z^{g-m_0-1}E(0,z)\theta_0}{\theta\big(\int_{\infty_+}^p {\rm d}{\bm v}+e\,|\,\Omega\big)}\right)=\sum_{j=1}^\infty \beta_j\frac{z^j}{j}, \qquad q(x)=\sum_{i,j=1}^\infty q_{ij}x_ix_j.
\end{gather*}

By a similar computation to \cite{N2} we have

\begin{Proposition} There exists a non-zero constant $c$ such that
\begin{gather*}
\tau\big(x;\tilde{\xi}(D)\big)= c \exp\left(\sum_{j=1}^\infty \beta_j x_j+\frac{1}{2}q(x)\right)\theta(\Gamma x+e\,|\,\Omega).
\end{gather*}
\end{Proposition}

By Proposition \ref{gap-sequence} the top term of the Schur function expansion of $\tau\big(x;\tilde{\xi}(D)\big)$ is determined. Let $\lambda$ be the partition defined by
\begin{gather*}
\lambda=(g,g-1,\dots ,m_0+1,m_0-1,\dots ,1,0)-(g-1,\dots ,1,0)=\big(1^{g-m_0}\big).
\end{gather*}
By \cite[Corollaries~1 and~2]{N3} the partition corresponding to the Schur function which appears in the top term of the expansion of $\tau\big(x;\tilde{\xi}(D)\big)$ is given by the conjugate partition of $\lambda$, ${}^t \lambda=(g-m_0)$. Taking the conjugate of $\lambda$ is due to the minus sign in the definition~(\ref{def-e}) of~$e$. By the form of the frame $\tilde{\xi}(D)$ the Schur function expansion of $\tau\big(x;\tilde{\xi}(D)\big)$ begins from $s_{(g-m_0)}(x)=p_{g-m_0}(x)$. Thus

\begin{Proposition}\label{tau-exp-nonsing} The following expansion holds,
\begin{gather*}
\tau\big(x;\tilde{\xi}(D)\big)=p_{g-m_0}(x)+\sum_{(g-m_0)<\mu}\xi_\mu s_\mu(x),
\end{gather*}
where $(g-m_0)<\mu$ means that, if $\mu=(\mu_1,\dots ,\mu_r)$, $g-m_0\leq \mu_1$ and $\mu\neq (g-m_0)$.
\end{Proposition}

\section{Degeneration}\label{section7}
We consider the limit $\lambda_{j+n}\rightarrow \lambda_j$, $j\in [n]$, of the curve (\ref{hyperelliptic}). By Corollary~\ref{basis-tilde-h} and Lemma~\ref{basis-change}, the limits $\xi^0(D)$ and $\tilde{\xi}^0(D)$ of $\xi(D)$ and $\tilde{\xi}(D)$ exist respectively. Since the Pl\"ucker coordinates of $\tilde{\xi}(D)$ tends to those of $\tilde{\xi}^0(D)$ the following equation holds
 \begin{gather*}
 \lim \tau\big(x;\tilde{\xi}(D)\big)=\tau\big(x;\tilde{\xi}^0(D)\big).
 \end{gather*}
In this section we show that $\xi^0(D)$ can be transformed to a frame of the form (\ref{soliton-frame}) by a gauge transformation.

The hyperelliptic curve (\ref{hyperelliptic}) tends to
\begin{gather*}
y^2=F(x)^2, \qquad F(x)=\prod_{j=1}^n(x-\lambda_j).
\end{gather*}

Let
\begin{gather*}
f(z)=\prod_{j=1}^n(1-\lambda_j z).
\end{gather*}
Then the Taylor series $y(z)$ of $y$ around $\infty_+$ tends to
\begin{gather*}
y(z)=z^{-n} f(z),
\end{gather*}
where we use the same symbol $y(z)$ for the limit of $y(z)$. Let $g_m^0$ and $f_m^0$ be the limits of~$g_m$ and~$f_m$ respectively. Then
\begin{gather*}
g_m^0=f^0_m(z)=z^{-m}f(z), \qquad m\geq n.
\end{gather*}

To determine the limit of $h_i$ we need to specify the limit of the point $p_i=(c_i,y_i)$. We do this in the following way.

Since $c_i$ does not depend on $\{\lambda_j\}$, $p_i$ goes to
\begin{gather}
p_i^0=(c_i,\varepsilon_i F(c_i)),\label{pi0}
\end{gather}
where $\varepsilon_i=\pm 1$. Let $k$ be an integer such that
\begin{gather*}
0\leq k\leq m_0.
\end{gather*}
Set
\begin{gather*}
l=m_0-k.
\end{gather*}
We assume, in (\ref{pi0}), that
\begin{gather*}
\varepsilon_i=
\begin{cases}
-1,&1\leq i\leq k,\\
\hphantom{-}1,& k+1\leq i\leq m_0.
\end{cases}
\end{gather*}
For simplicity set
\begin{gather*}
d_i=c_{k+i},\qquad i\in [l].
\end{gather*}
Then
\begin{gather*}
p^0_i = (c_i, -F(c_i)), \qquad i\in[k],\qquad p^0_{k+i}=(d_i, F(d_i)),\qquad i\in[l].
\end{gather*}
This condition is satisfied if $p_1,\dots ,p_k$ are in a small neighborhood of $\infty_-$ and $p_{k+1},\dots ,p_{m_0}$ are in a small neighborhood of~$\infty_+$.

The limits of the quantities in the numerator of (\ref{hj}) are
\begin{gather*}
-y_i+g_n(c_i) \longrightarrow 2 F(c_i), \qquad i\in[k],\qquad -y_{k+i}+g_n(d_i)\longrightarrow 0, \qquad i\in[l],\\
y+g_n(x)\longrightarrow 2z^{-n}f(z).
\end{gather*}
Therefore the limit $h_i^0(z)$ of $h_i$ becomes
\begin{gather*}
h^0_i(z) = z^{-(n-1)}\frac{f(z)-z^nF(c_i)}{1-c_iz},\qquad i\in[k],\\
h^0_{k+i}(z)=z^{-(n-1)}\frac{f(z)}{1-d_iz},\qquad i\in[l].
\end{gather*}

Then the frame $\xi^0(D)$ is given by
\begin{gather}
\xi^0(D)=\big(\ldots,z^{-m_0-2}f(z),z^{-m_0-1}f(z),z^{g-m_0}h^0_{m_0}(z),\dots ,z^{g-m_0}h^0_{1}(z),z^{g-m_0}\big).\label{xi-D-0}
\end{gather}

\begin{Definition}We denote the point of UGM corresponding to the frame (\ref{xi-D-0}) by $U^0(D)$.
\end{Definition}

In order to identify the solution corresponding to $U^0(D)$ with a soliton solution we change a~basis and make a gauge transformation. To this end let
\begin{gather*}
\varphi(z)=\prod_{j=1}^l(1-d_jz), \qquad \varphi_i(z)=\frac{\varphi(z)}{1-d_iz}.
\end{gather*}

Consider the gauge transformation $\varphi(z) U^0(D)$ of $U^0(D)$. Then the following set of functions is a basis of $\varphi(z) U^0(D)$,
\begin{gather*}
z^{g-m_0}\varphi(z), \qquad z^{g-m_0}\varphi(z) h_i^0(z),\qquad i\in[k],\qquad z^{-m_0}\varphi_i(z) f(z),\qquad i\in[l],\nonumber\\
z^{-i}\varphi(z) f(z),\qquad i\geq m_0+1. 
\end{gather*}

\begin{Lemma}\label{phi-U0}The following set of functions is a basis of $\varphi(z) U^0(D)$:
\begin{gather*}
z^{g-m_0}\varphi(z), \qquad z^{g-m_0}\varphi(z)h_i^0(z), \qquad i\in[k],\qquad z^{-i}f(z), \qquad i\geq k+1.
\end{gather*}
\end{Lemma}
\begin{proof} Since $d_i$'s are mutually distinct and $\deg \varphi_i(z)=l-1$,
\begin{gather*}
\operatorname{Span}_{\mathbb C}\{\varphi_i(z)\,|\, i\in[l]\} = \operatorname{Span}_{\mathbb C}\big\{z^{i}\,|\,0\leq i\leq l-1\big\},
\end{gather*}
where $\operatorname{Span}_{\mathbb C}\{\ast\}$ denotes the vector space spanned by $\{\ast\}$. Therefore, noticing that $\deg \varphi(z)=l$, we have
\begin{gather*}
\operatorname{Span}_{\mathbb C}\big\{ z^{-m_0}\varphi_i(z), \, i\in [l], \, z^{-i}\varphi(z),\, i\geq m_0+1 \big\} =\operatorname{Span}_{\mathbb C}\big\{z^{-i}\,|\,i\geq k+1\big\},
\end{gather*}
which shows the lemma.
\end{proof}

Define $b_{i,j}$ by
\begin{gather}
z^g h_i^0(z)=\frac{f(z)-z^nF(c_i)}{1-c_iz}=\sum_{j=0}^{n-1}b_{i,j} z^j.\label{hi0}
\end{gather}

It is possible to erase the terms of degree less than $l$ in $\varphi(z)\sum\limits_{j=0}^{n-1}b_{i,j}z^j$ by subtracting an appropriate linear combination of $z^rf(z)$, $0\leq r\leq l-1$. It means that there exist constants $\eta_r$, $0\leq r\leq l-1$ and a unique polynomial $G_i(z)$ of degree at most $n-1$ satisfying the following equation:
\begin{gather}
\varphi(z)\sum_{j=0}^{n-1}b_{i,j}z^j=\sum_{r=0}^{l-1}\eta_r z^rf(z)+z^lG_i(z).\label{def-gi}
\end{gather}

\begin{Proposition}\label{basis-UD0} A basis of $\varphi(z) U^0(D)$ is given by
\begin{gather*}
z^{g-m_0}\varphi(z), \qquad z^{-k}G_i(z), \qquad i \in [k], \qquad z^{-i}f(z), \qquad i\geq k+1.
\end{gather*}
\end{Proposition}

\begin{proof} Multiplying (\ref{def-gi}) by $z^{-m_0}$ we have
\begin{gather}
z^{-m_0}\varphi(z)\sum_{j=0}^{n-1}b_{i,j}z^j=\sum_{j=k+1}^{m_0}\eta_{m_0-j} z^{-j}f(z)+z^{-k}G_i(z).\label{def-gi-1}
\end{gather}
Then the lemma follows from Lemma \ref{phi-U0}, (\ref{hi0}), (\ref{def-gi-1}).
\end{proof}

Define
\begin{gather}
{\widehat U}^0(D)=\varphi(z)f(z)^{-1}U^0(D).\label{gauge-trf}
\end{gather}
Then

\begin{Theorem}\label{main-th}The following set of functions is a basis of ${\widehat U}^0(D)$:
\begin{gather*}
z^{-k}\sum_{i=1}^n\frac{a_{i,j}}{1-\lambda_i z}, \qquad j\in [k+1], \qquad z^{-i}, \qquad i\geq k+1,
\end{gather*}
where $A=(a_{i,j})_{i\in [n],j\in [k+1]}$ is given by
\begin{gather*}
A=\left[
\begin{matrix}
\frac{D_1}{\Lambda_1}C_{1,1}&\ldots&\frac{D_1}{\Lambda_1}C_{1,k}&\frac{D_1}{\Lambda_1}\\
\vdots&\quad&\vdots&\vdots\\
\frac{D_n}{\Lambda_n}C_{n,1}&\ldots&\frac{D_n}{\Lambda_n}C_{n,k}&\frac{D_n}{\Lambda_n}\\
\end{matrix}
\right],\\
\Lambda_i=\prod_{r\neq i}^n(\lambda_i-\lambda_r),\qquad D_i=\prod_{s=1}^l(\lambda_i-d_s), \qquad C_{i,j}=\prod_{r\neq i}^n(c_j-\lambda_r).
\end{gather*}
\end{Theorem}

Define
\begin{gather*}
\widehat{H}_j=
\begin{cases}
\displaystyle z^{-k}\sum_{i=1}^n\frac{a_{i,j}}{1-\lambda_i z}, & j\in [k+1],\\
z^{-j+1}, &j\geq k+2,
\end{cases}\\
\widehat{\xi}^0(D)=\big[\ldots,\widehat{H}_3,\widehat{H}_2,\widehat{H}_1\big].
\end{gather*}

Then Theorem \ref{main-th} tells that $\tau\big(x;\widehat{\xi}^0(D)\big)$ is an $(n,k+1)$ soliton.

\begin{Remark}\label{remark2} In \cite{Abenda} the case $m_0=n-1$, which corresponds to regular solutions, is studied. In this case the matrix $A$ in Theorem~\ref{main-th} is the dual matrix of that in~\cite{Abenda}. More precisely the correspondence is as follows. We denote ${\rm Gr}(M,N)$ the Grassmannian which is the set of $M$-dimensional subspaces in an $N$-dimensional vector space. Let $k$, $l$ be as in the present paper. Consider the $l\times n$ matrix $B$ in ${\rm Gr}(l,n)$ corresponding to the divisor $\big(\gamma_1^{(l)},\dots ,\gamma_l^{(l)}, \delta_1^{(l)},\dots ,\delta_k^{(l)}\big)$ given in \cite[Theorem~9.1]{Abenda}. We identify $d_i=\gamma_i^{(l)}$, $c_i=\delta_i^{(l)}$, $\lambda_j=\kappa_j$. Since $n-l=k+1$, the dual matrix $\widehat{B}$ of $B$ given in \cite[Section~10]{Abenda} belongs to ${\rm Gr}(k+1,n)$. Then it can be shown that ${}^t A$ and $\widehat{B}$ give the same element of ${\rm Gr}(k+1,n)$. That is there exists an invertible $(k+1)\times (k+1)$ matrix $C$ such that ${}^t AC=\widehat{B}$. So the tau function corresponding to~$A$ in Theorem~\ref{main-th} coincides with that corresponding to $\widehat{B}$ up to constant multiple.
\end{Remark}

The theorem is proved by expanding elements of the basis in Proposition \ref{basis-UD0} into partial fractions by using the following lemma which easily follows from the definition (\ref{def-gi}) of $G_i(z)$.

\begin{Lemma} For $1\leq j\leq n$ we have
\begin{gather*}
G_i\big(\lambda_j^{-1}\big)=\lambda_j^{-(n-1)}\prod_{s=1}^l(\lambda_j-d_s)\prod_{r\neq j}^n(c_i-\lambda_r).
\end{gather*}
\end{Lemma}

\section[The tau function corresponding to $\tilde{\xi}^0(D)$]{The tau function corresponding to $\boldsymbol{\tilde{\xi}^0(D)}$}\label{section8}
In this section we compute the tau function $\tau\big(x;\tilde{\xi}^0(D)\big)$ and the corresponding adjoint wave function.

By taking the limit of (\ref{relation-tau}) we have
\begin{gather}
\tau(x;\xi^0(D))=\left(\prod_{1\leq i<j\leq m_0}(c_j-c_i)\right) \tau\big(x;\tilde{\xi}^0(D)\big). \label{relation-tau-1}
\end{gather}
The tau function $\tau\big(x;\widehat{\xi}^0(D)\big)$ can be expressed by $\tau\big(x;\xi^0(D)\big)$ as follows.

If two frames of points of UGM are related by
\begin{gather*}
\xi'=\psi(z) \xi, \qquad \psi(z)=1+O(z), \qquad \log \psi(z)=\sum_{i=1}^\infty g_i\frac{z^i}{i},
\end{gather*}
then, by (\ref{Psi-tau-relation}),
\begin{gather*}
\tau(x;\xi')=\tau(x;\xi)\exp\left(\sum_{i=1}^\infty g_i x_i\right).
\end{gather*}

So in our case we need to compute the expansion of $\log f(z)$ and $\log \varphi(z)$. Let $P_i(u_1,\dots ,u_r)$ be the power sum symmetric function defined by
\begin{gather*}
P_i(u_1,\dots ,u_r)=\sum_{j=1}^r u_j^i.
\end{gather*}
Then
\begin{gather*}
\log f(z)^{-1}=\sum_{j=1}^\infty P_j(\lambda_1,\dots ,\lambda_n)\frac{z^j}{j},\qquad \log \varphi(z)=-\sum_{j=1}^\infty P_j(d_1,\dots ,d_l) \frac{z^j}{j}.
\end{gather*}

By (\ref{gauge-trf}) we have
\begin{gather}
\tau\big(x;\widehat{\xi}^0(D)\big)= \exp\left(\sum_{j=1}^\infty(P_j(\lambda_1,\ldots,\lambda_n)-P_j(d_1,\ldots,d_n)) x_j \right)\tau\big(x;\xi^0(D)\big).\label{tau-gauge-1}
\end{gather}

By (\ref{relation-tau-1}) and (\ref{tau-gauge-1})
\begin{gather}
\tau\big(x;\tilde{\xi}^0(D)\big)=\left(\prod_{1\leq i<j\leq m_0}(c_j-c_i)\right)^{-1}\exp\left(\sum_{j=1}^\infty(-P_j(\lambda_1,\ldots,\lambda_n)+P_j(d_1,\ldots,d_n))x_j \right)\nonumber\\
\hphantom{\tau\big(x;\tilde{\xi}^0(D)\big)=}{} \times \tau\big(x;\widehat{\xi}^0(D)\big).\label{relation-tau-2}
\end{gather}

The tau function corresponding to $\widehat{\xi}^0(D)$ can be computed by (\ref{nk-soliton}) with the matrix $A$ given in Theorem~\ref{main-th}. Let us compute $A_I$.

\begin{Lemma}\label{AI}For $I\in \binom{[n]}{k+1}$ we have
\begin{gather}
A_I=\Xi \times \prod_{i\in I}\left(\frac{\prod\limits_{j=1}^l(\lambda_i-d_j)}{\prod\limits_{j=1}^k(\lambda_i-c_j)}
\frac{1}{\prod\limits_{r\neq i}^n(\lambda_i-\lambda_r)}\right)
\prod_{i,j\in I, i<j}(\lambda_i-\lambda_j),\label{AI-1}\\
\Xi=\prod_{i<j}^k(c_j-c_i)\prod_{j=1}^k\prod_{r=1}^n(c_j-\lambda_r).\label{Xi-1}
\end{gather}
\end{Lemma}

The lemma can be proved using the following Cauchy like formula which is easily proved:
\begin{gather*}
\left|\begin{matrix}
\frac{1}{c_1-\lambda_1}&\ldots&\frac{1}{c_k-\lambda_1}&1\\
\vdots&\quad&\vdots&\vdots\\
\frac{1}{c_1-\lambda_{k+1}}&\ldots&\frac{1}{c_k-\lambda_{k+1}}&1\\
\end{matrix}
\right|
=\frac{\prod\limits_{i<j}^k(c_j-c_i)\prod\limits_{i<j}^{k+1}(\lambda_i-\lambda_j)} {\prod\limits_{i=1}^k\prod\limits_{j=1}^{k+1}(c_i-\lambda_j)}.
\end{gather*}

We assign weight $i$ to $x_i$. Then

\begin{Corollary}\label{limit-tau-formula}\quad
\begin{enumerate}\itemsep=0pt
\item[$(i)$] We have
\begin{gather}
\tau\big(x;\tilde{\xi}^0(D)\big) = c'\exp\left(\sum_{i=1}^\infty(-P_i(\lambda_1,\dots ,\lambda_n)+P_i(d_1,\dots ,d_l))x_i\right)\nonumber\\
\hphantom{\tau\big(x;\tilde{\xi}^0(D)\big) =}{} \times
\sum_{I\in \binom{[n]}{k+1}} \prod_{i\in I}\left(\frac{\prod\limits_{j=1}^l(\lambda_i-d_j)}{\prod\limits_{j=1}^k(\lambda_i-c_j)}
\frac{1}{\prod\limits_{r\neq i}^n(\lambda_i-\lambda_r)}\right)\exp\left(\sum_{i\in I}\eta_i\right),\label{formula-final}
\end{gather}
where
\begin{gather*}
c'=\frac{\Xi}{\prod\limits_{i<j}^{m_0}(c_j-c_i)},
\end{gather*}
and $\Xi$ is given by \eqref{Xi-1}.

\item[$(ii)$] The following expansion holds,
\begin{gather*}
\tau\big(x;\tilde{\xi}^0(D)\big)=p_{g-m_0}(x)+\cdots,
\end{gather*}
where $\cdots$ part contains only terms with the weights greater than $g-m_0$.
\end{enumerate}
\end{Corollary}

\begin{proof} (i) It follows from Theorems \ref{main-th} and~\ref{S-frame}, Lemma \ref{AI}, (\ref{nk-soliton}), (\ref{relation-tau-2}).

(ii) It follows from Proposition \ref{tau-exp-nonsing}.
\end{proof}

Next we examine the conditions of the regularity of $\tau\big(x;\tilde{\xi}^0(D)\big)$. Hereafter we assume that $\lambda_i$, $c_i$, $d_i$, $x_i$ are real for all possible $i$ and that
\begin{gather*}
\lambda_1<\cdots<\lambda_n.
\end{gather*}
By Corollary \ref{limit-tau-formula}(ii) if $m_0<g$, $\tau\big(x;\tilde{\xi}^0(D)\big)$ becomes singular since it vanishes at $x={}^t(0,0,\ldots)$. So let us consider the case $m_0=g$. Then $k+l=g=n-1$. In this case the tau function~(\ref{formula-final}) is the same as that studied in~\cite{Abenda} as mentioned in Remark~\ref{remark2}. Thus the following proposition is proved in~\cite{Abenda}.

\begin{Proposition}\label{positive}
If there exists a permutation $w$ of $\{1,2,\dots ,n\}$ such that
\begin{gather*}
\lambda_1<c_{w(1)}<\lambda_2<c_{w(2)}<\cdots<c_{w(n-1)}<\lambda_n,
\end{gather*}
then the sign of $A_I$ does not depend on $I\in \binom{[n]}{k+1}$.
\end{Proposition}

Proposition \ref{positive} means that $\tau\big(x;\tilde{\xi}^0(D)\big)$ is positive for all $x_1,x_2,\dots$, if it is multiplied by some constant and the solution $u=2\partial_{x}^2 \log \tau\big(x;\tilde{\xi}^0(D)\big)$ of the KP equation~(\ref{KP-equation}) has no singularity.

Finally we give explicitly the adjoint wave function, which we denote by $\Psi^0(x;z)$, correspon\-ding to the tau function in Corollary~\ref{limit-tau-formula}. The result is
\begin{gather*}
\Psi^0(x;z)=\frac{1}{ \Phi(x) } \frac{ \sum\limits_I \Delta_I A_I \left(\prod\limits_{i\in I^c}(1-\lambda_i z)\right){\rm e}^{\sum\limits_{i\in I} \eta_i} }{ \prod\limits_{j=1}^l(1-d_jz) }
\left(-\sum\limits_{i=1}^\infty x_i z^{-i}\right),
\end{gather*}
where $I^c$ denotes the complement of $I$ in $[n]$, $\Delta_I=\Delta_I(\lambda_1,\dots ,\lambda_n)$, $A_I$ is given by~(\ref{AI-1}) and $\Phi(x)$ is the part of $\tau\big(x;\tilde{\xi}^0(D)\big)$ which is obtained by removing the part in front of the sum symbol (the constant and the exponential function). The function $\Psi^0(x;z)$ is the expression of the limit of $\Psi(x;z)$ near $\infty_+$.

Notice that the poles at $p_1,\dots ,p_k$ of $\Psi(x;z)$ disappear in the limit. This is possible because we consider the reducible degeneration of the curve $X$.

\subsection*{Acknowledgements}
The author would like to thank Simonetta Abenda, Yasuhiko Yamada for useful discussions. He is also grateful to Kanehisa Takasaki for comments on the manuscript of the paper and to Yuji Kodama for many inspiring questions and comments. This work is supported by JSPS Grants-in-Aid for Scientific Research No.15K04907.

\pdfbookmark[1]{References}{ref}
\LastPageEnding

\end{document}